\documentclass[11pt,onecolumn,draftcls]{IEEEtran}
\usepackage{amsmath,amssymb,amsthm}    % need for subequations
\usepackage[dvips]{graphicx}   % need for figures
\usepackage{verbatim}   % useful for program listings
\usepackage{color}      % use if color is used in text
\usepackage{hyperref}   % use for hypertext links, including those to external documents and URLs
\usepackage{epsfig}
\usepackage{float}%\usepackage{keywords}

\usepackage{cite}
\usepackage[font=footnotesize]{subfig}

\newcommand{\bi}{\begin{itemize}}
\newcommand{\ei}{\end{itemize}}
\newcommand{\beq}{\begin{equation}}
\newcommand{\eeq}{\end{equation}}
\newcommand{\bqn}{\begin{eqnarray*}}
\newcommand{\eqn}{\end{eqnarray*}}
\newcommand{\ba}{\begin{array}}
\newcommand{\ea}{\end{array}}
\newcommand{\bs}{\begin{small}}
\newcommand{\es}{\end{small}}

\newtheorem{theorem}{Theorem}[section]

\newtheorem{proposition}[theorem]{Proposition}

\newenvironment{definition}[1][Definition]{\begin{trivlist}
\item[\hskip \labelsep {\bfseries #1}]}{\end{trivlist}}

\begin{document}

\title{False Discovery Rate Based Distributed Detection in the Presence of Byzantines}
\author{Aditya~Vempaty*,~\IEEEmembership{Student Member,~IEEE}, Priyadip Ray,~\IEEEmembership{Member,~IEEE}, Pramod~K.~Varshney,~\IEEEmembership{Fellow,~IEEE}%
%\thanks{EDICS: SEN-LOCL, SSP-PARE, SSP-APPL, SEN-FUSE}
\thanks{A. Vempaty and P. K. Varshney are with Department of EECS, Syracuse University, Syracuse, NY 13244. P. Ray is with G. S. Sanyal School of Telecommunications, Indian Institute of Technology, Kharagpur, West Bengal, India.
Email:  \{avempaty, varshney\}@syr.edu, priyadip.ray@gssst.iitkgp.ernet.in}}
\markboth{IEEE TRANSACTIONS ON AEROSPACE AND ELECTRONIC SYSTEMS (DRAFT)}{}

\date{}

\maketitle

\begin{abstract}
Recent literature has shown that the control of False Discovery Rate (FDR) for distributed detection in wireless sensor networks (WSNs) can provide substantial improvement in detection performance over conventional design methodologies. In this paper, we further investigate system design issues in FDR based distributed detection. We demonstrate that improved system design may be achieved by employing the Kolmogorov-Smirnov distance metric instead of the deflection coefficient, as originally proposed in \cite{ray_aes11}. We also analyze the performance of FDR based distributed detection in the presence of Byzantines.  Byzantines are malicious sensors which send falsified information to the Fusion Center (FC) to deteriorate system performance. We provide analytical and simulation results on the global detection probability as a function of the fraction of Byzantines in the network. It is observed that the detection performance degrades considerably when the fraction of Byzantines is large. Hence, we propose an adaptive algorithm at the FC which learns the Byzantines' behavior over time and  changes the FDR parameter to overcome the loss in detection performance. Detailed simulation results are provided to demonstrate the robustness of the proposed adaptive algorithm to Byzantine attacks in WSNs. 
\end{abstract}


\section{Introduction}
\label{intro}
In recent years, wireless sensor networks (WSNs) have been employed extensively to monitor a region of interest (ROI) for reliable detection/estimation/tracking of events \cite{akyildiz_commag02,niu_tsp06,Wang_TMC12,Veeravalli&Varshney_12}. In this paper, we focus on distributed target detection in WSNs, which has been a very active area of research in the recent past. In distributed detection \cite{Varshney:book}, due to power and bandwidth constraints, each sensor, instead of sending its raw data, sends quantized data (local decision) to a central observer or Fusion Center (FC). The FC combines these local decisions based on a fusion rule to come up with a global decision. There has been extensive research on distributed detection with fusion of local decisions \cite{Varshney:book}. Optimum fusion rules have been derived for the distributed detection problem under various assumptions \cite{Chair&Varshney:86AES,Viswanathan&Varshney:97Proc,Tsitsiklis:bookchapter,Kam-etal:AES1992,Drakopoulos&Lee:AES1991}. Most of these fusion rules require complete knowledge of the local sensor performance metrics, such as the probability of detection and false alarm. However, in large wireless sensor networks and under complex target signal models, the local sensor performance metrics may not be known or may be very  difficult to estimate. To address the scenario of unknown local sensor performance metrics, in  \cite{Niu&Varshney:05eurasip} \cite{Niu&etal:06ijif}, the authors have proposed employing the total number of detections (also referred to as the \textit{count statistic}) as a decision statistic at the FC. The fusion rule based on the count statistic leads to a decision rule where the sensor decisions are weighed equally, even though the SNR at each sensor may be different.

In general, obtaining the optimal local decision rules has been shown to be a very difficult problem \cite{Varshney:book} \cite{Willett-etal:00SP}. Under the conditional independence assumption, it has been shown that the use of identical local decision rules  is optimal under asymptotic conditions (i.e., the number of sensors $N \rightarrow \infty$) \cite{Tsitsiklis88}. Although the optimality of identical decision rules does not hold in general \cite{Tsitsiklis88} \cite{Varshney:book} \cite{Willett-etal:00SP}, design of  non-identical decision rules is computationally very complex and researchers have generally employed identical decision rules based on asymptotic optimality of identical decision rules. Recently, False Discovery Rate (FDR) based distributed detection has been proposed by Ermis and Saligrama \cite{ermis&saligrama10} and Ray and Varshney \cite{ray_ijdsn08} \cite{ray_aes11}. In \cite{ray_aes11}, the authors have proposed a scheme for distributed detection in WSNs based on the control of FDR. It has been shown  that under the assumption that the FC employs a test statistic which is linear in count (count here refers to the total number of detections) to reach the global decision, control of the FDR leads to non-identical local decision rules and provides significant improvement over the system with identical decision rules. This scheme provides significant improvement in the global detection performance \cite{ray_aes11}. In \cite{ray_aes11}, the authors suggest maximization of the deflection coefficient to obtain the FDR design parameter. However, as the count statistic is non-Gaussian in general, maximization of the deflection coefficient does not guarantee  optimal global performance \cite{pincinbono95}.  In this paper, we further consider the problem of FDR based distributed detection and demonstrate that improved system performance may be achieved via optimization of the Kolmogorov-Smirnov distance instead of the deflection coefficient. 

Most of the  research in the field of distributed detection has been carried out under the assumption of a secure network. Only in the recent past, researchers have investigated the problem of security threats\cite{HKD06} on sensor  networks. In this paper, we focus on one particular class of security attacks, known as the Byzantine attack\cite{marano_tsp09,rawat_tsp11,adds,Vempaty_TSP} (also referred to as the Data Falsification Attack). Byzantine attack involves malicious sensors within the network which send false information to the FC to disrupt the global decision making process. Byzantines intend to deteriorate the detection performance of the network by suitably modifying their decisions before transmission to the FC. Marano et al.\cite{marano_tsp09} first considered the distributed detection problem in the presence of Byzantines under the assumption that the Byzantines have perfect knowledge of the underlying true hypothesis. In \cite{marano_tsp09}, the authors also presented the optimal attacking distributions for the Byzantines such that the detection error exponent is minimized at the FC. Rawat et al. in \cite{rawat_tsp11} considered the more practical case when the Byzantines did not have the knowledge of the underlying true hypothesis and also proposed a simple algorithm at the FC to identify the Byzantines in the network. In \cite{gagrani_allerton11}, stochastic resonance \cite{HaoChen_tsp07} \cite{HaoChen_tsp08} was employed to mitigate the effect of Byzantines in a distributed inference network. In our previous work \cite{Vempaty_TSP}, we have also analyzed localization in WSNs in the presence of Byzantines and proposed mitigation techniques to make the Byzantines `ineffective'. In this paper, we study the performance of the FDR based distributed detection framework in the presence of Byzantines. It is observed that the global detection performance deteriorates  rapidly when the fraction of Byzantine sensors increases. Hence, we propose a novel algorithm at the FC, based on a modified Kolmogorov goodness-of-fit test, which detects the fraction of Byzantines present in the network and adaptively changes the FDR parameter to improve the detection performance. 

The key contributions of this paper are summarized as follows:
\bi 
\item We propose maximization of the Kolmogorov-Smirnov distance instead of the deflection coefficient to obtain the FDR design parameter  and demonstrate that it considerably improves system performance. 
\item We define a Byzantine attack model and show that the FDR value is controlled even in the presence of Byzantines; however the local sensor detection performance deteriorates considerably when the fraction of Byzantines is large. 
\item We next study the performance of FDR based distributed detection in the presence of Byzantine attacks and provide analytical and simulation results on the effect of Byzantines on global detection performance.
\item Finally, we propose an  algorithm  which adaptively changes the system  parameters by learning the Byzantines' behavior over time and demonstrate that the proposed algorithm  provides improved system performance in the presence of Byzantines.
\ei

The remainder of the paper is organized as follows: In Section \ref{prel}, we introduce the system model and the assumptions made in the paper. We also formally define False Discovery Rate (FDR) and briefly discuss the FDR based distributed detection scheme proposed in \cite{ray_aes11}. We propose some changes to the system design algorithm proposed in \cite{ray_aes11} and show the improvement in system performance in Section \ref{DM}. In Section \ref{Byz}, we show the performance degradation of FDR based schemes in the presence of Byzantines. We show that although the FDR value is maintained at the specified value, the power of the test reduces. We provide analytical results on the performance of FDR based distributed detection in the presence of Byzantines in Section \ref{perf}. We propose an approximation to the optimal parameter design approach which is computationally efficient and building on it, propose the adaptive distributed detection scheme in Section \ref{adap}. We conclude with a discussion on possible future directions in Section \ref{conc}.

\section{Preliminaries}
\label{prel}
\subsection{System Model}
\label{sys}
We consider a parallel fusion network where $N$ sensors are randomly deployed in the Region of Interest (ROI). Each sensor receives noisy target signals, $s_i$ (for $i=1,2,...,N$) and makes a decision $b_i$ regarding the presence/absence of the target, which is then transmitted to the FC. The FC makes a global decision ($b_0\in \{1,0\}$) regarding the presence/absence of the target using the transmitted local decisions $\left\{b_i\right\}_{i=1,\ldots,N}$. We assume that the channels between the local sensors and the FC are ideal (for results on distributed detection with imperfect channels, see \cite{chen&etal:06spm}, \cite{Gini-etal:98IEE}, and references therein). We consider the presence of $M=\alpha N$ ($0\leq\alpha\leq 1$) Byzantines in the network. These Byzantines' aim is to send falsified information to the FC and deteriorate the detection performance. Their model and attack strategy would be described later in Section \ref{Byz}.

As discussed earlier, due to the bandwidth and energy constraints, each local sensor sends a binary decision (0/1) to the FC based on a local hypothesis test. The local sensor's hypothesis test can be formulated as follows:
\begin{align}
\label{hyp}
H_0:	s_i &=n_i:	\text{Target absent}\\
H_1:	s_i &=a_i+n_i:	\text{Target present}
\end{align}
where $a_i=\sqrt{P_i}$ is the signal amplitude received at the $i^{th}$ sensor due to the presence of the target and $n_i \in \mathcal{N}(0,1)$ represent i.i.d. Gaussian noise. In this paper, we assume that the signal power received due to that emitted by the target drops to zero outside its finite radius of influence ($d_0$). The general signal model used is
\begin{equation}
P_i=g(d_i)
\end{equation}
where $P_i$ is the signal power received at the $i^{th}$ sensor which is at a distance $d_i$ from the target. As discussed above, the following $g(\cdot)$ has been used in this paper:
\begin{equation}
g(x)=
\begin{cases}
P_0,	& \text{if $0 \leq x \leq d_0$}\\
0,		& \text{if $x > d_0$}
\end{cases}
\end{equation}
The above model is adopted primarily for analytical convenience; however the results provided in this paper may be easily extended to more complex target signal models, such as where the target signal decays exponentially or in an inverse square manner with distance. 

The FC makes a global decision based on the vector of local decisions $\underline{b}=(b_1,\cdots,b_N)$ received from all the sensors. The binary hypothesis testing problem at the FC is
\begin{eqnarray}
G_0:	P(\underline{b};G_0):	\text{Target absent}\\
G_1:	P(\underline{b};G_1):	\text{Target present}
\end{eqnarray}
where, $P(\underline{b};G_0)$ and $P(\underline{b};G_1)$ are the distributions of $\underline{b}$ in the absence of the target ($G_0$) in the ROI and in the presence of the target ($G_1$) in the ROI respectively.

Conventionally, for the distributed detection problem, identical decision rules are used at the local sensors. However in  \cite{ray_aes11}, FDR based non-identical sensor decision rules have been proposed and  shown to be superior to identical decision rules. In the next subsections, we introduce the concept of FDR and provide a brief description of FDR based distributed detection.

\subsection{False Discovery Rate (FDR)}
\label{FDR}
In statistical hypothesis testing, when a family of tests (e.g., multiple binary hypothesis tests) are conducted, it is often meaningful to define an  error rate for the family of tests instead for an individual test. Family wide error rate (FWER) \cite{lehman08} is perhaps the most popular error rate used in the literature. It is defined as the probability of committing any type I error or false alarm. If the error rate for each test is $\beta$ then the FWER $\beta_F$ for $k$ tests is 
\begin{equation}
\label{beta_F}
\beta_F=P(F\geq 1)=1-(1-\beta)^k
\end{equation}
where $F$ is the total number of false alarms. As can be seen from (\ref{beta_F}), as the number of tests $k$ increases, $\beta$ remains constant but $\beta_F$ increases. This is a fundamental problem in Multiple Comparison Procedures (MCP) and classical comparison procedures aim to control this error measure. Bonferroni procedure \cite{Simes86} is a widely employed procedure to control the FWER at a desired rate, but it results in significantly reduced power (probability of detection). A radically different and more liberal approach proposed by Benjamini and Hochberg \cite{benj95} controls the FDR, defined as the fraction of false rejections among those hypotheses rejected. Formally, FDR is defined as the expected ratio of the number of false alarms (declared $H_1$ when $H_0$ is true) to the total number of detections ($H_1$ declarations consisting of both true and false detections). 

\begin{center}
\begin{table}[htb]
\caption{Notations to define FDR}
\begin{tabular}{|p{1.5in}|p{1.0in}p{1.0in}|p{1.0in}|}
\hline
	& Declared $H_0$ & Declared $H_1$ & Total\\
   \hline
   $H_0$ true & $W$  & $F$ & $N_0$\\
   %\hline 
   $H_1$ true & $T$ & $S$ & $N-N_0$\\ 
   \hline
   Total & $N-R$ & $R$ & $N$\\
   \hline
\end{tabular}
%\label{FDR_table}
\end{table}
\end{center}

From Table I, the ratio of false alarms to the total number of detections can be viewed as the random variable,
\begin{equation}
A=
\begin{cases}
\frac{F}{F+S}	& \text{if $F+S\neq0$}\\
0	& \text{if $F+S=0$}
\end{cases}
\end{equation}

FDR is defined to be the expectation of $A$,
\begin{equation}
FDR=E(A)
\end{equation}

This metric was proposed by Benjamini and Hochberg \cite{benj95} along with the following centralized algorithm to control FDR for multiple comparisons.

\paragraph*{Algorithm to control FDR}
Suppose $p_1,p_2,\cdots, p_N$ are the p-values for $N$ tests and $p_{(1)}, p_{(2)},\cdots,p_{(N)}$ denote the ordered p-values. The p-value for an observation $s_i$ is defined as
\begin{equation}
p_i =\int_{s_i}^\infty f_0(s)ds
\end{equation}
where, $f_0(s)$ is the probability density function of the observation under $H_0$. The algorithm by Benjamini and Hochberg \cite{benj95} which keeps the FDR below a value $\gamma$, is provided
below.
\bi
\item[1.] Calculate the p-values of all the observations and arrange them in ascending order.
\item[2.] Let $d$ be the largest $i$ for which $p_{(i)}\leq i\gamma/N$.
\item[3.] Declare all observations corresponding to $p_{(i)}$, $i =1,\cdots ,d$, as $H_1$.
\ei

Under the assumption of independence of test statistics corresponding to the true null hypotheses ($H_0$), this procedure controls the FDR at $\gamma$. Note that the FDR based decision making system looks for the largest index $i = d$ such that $p_{(d)}\leq d\gamma/N$. There may be other indices $i = l$, where $l < d$ for which the condition $p_{(l)}\leq l\gamma/N$ may be true, but the FDR based decision system looks for the largest value of $i$ for which this is true. The reason behind this, as discussed in \cite{benj95} and subsequently pointed out in \cite{ray_aes11}, is to achieve the largest probability of detection while constraining the FDR to less than or equal to $\gamma$. Further discussion including the proof of this algorithm is omitted for the brevity of the paper and may be found in \cite{benj95}.

Application of the above algorithm for a multi-sensor signal detection framework requires the centralized ordering of p-values at the FC. As discussed earlier, for a distributed detection problem, due to the energy and communication constraints, local sensors send quantized information, and often only send one bit to the FC. Hence, a distributed ordering scheme of the p-values is necessary. Ray \& Varshney \cite{ray_aes11} have proposed a decentralized FDR procedure based on a feedback mechanism which in effect orders p-values at the sensors. This distributed algorithm achieves the same performance as the Benjamini and Hochberg centralized procedure. The algorithm and further discussion may be found in \cite{ray_aes11}. 

%Under the assumption that each sensor can broadcast a single bit to the FC and all other sensors in the network can hear, the decentralized FDR control algorithm is provided below \cite{ray_aes11}:
%\bi
%\item[1.] Each sensor compares the p-value of its observation $p_i$ to $\gamma$.
%\item[2.] The $l$ sensors having $p_i>l$, broadcast this decision using a single bit to the FC and to the entire network.
%\item[3.] All the sensors update their thresholds to $\gamma((N-l)/N)$.
%\item[4.] Only those sensors that did not broadcast their decisions earlier perform the threshold test. Let $k$ be the number of sensors that have p-values greater than the present threshold. These $k$ sensors broadcast their decision using a single bit to the FC and the entire network.
%\item[5.] All the sensors update their thresholds to $\gamma((N-l-k)/N)$.
%\item[6.] Steps 4 and 5 are repeated until there are no more sensors reporting p-values greater than the current threshold.
%\item[7.] If $C=l+k+\cdots$, is the total number of sensors which broadcast their decisions to the FC, the total number of decisions which are classified as $H_1$ is given by $N-C$.
%\ei

\section{FDR based Distributed Detection}
\label{DM}

%\textbf{we have not really introduced gamma and tau and T yet properly. Do we need to do that?}
In this section, we summarize the design guidelines for FDR based distributed detection proposed in \cite{ray_aes11} and show that system design needs to be re-visited. The design parameters in a FDR based distributed detection system \cite{ray_aes11} are $\gamma$ and $T_{FDR}$, where $\gamma$ is the FDR parameter and $T_{FDR}$ is the global threshold. %For the sake of comparison, we also study system design for an identical threshold scheme, where the design parameters are $\tau$ and $T_{IT}$, where $\tau$ is the local observation threshold parameter ($Q(\tau)=p_{fa}$ is the threshold on the p-values) and $T_{IT}$ is the global threshold. 
The system-wide probability of false alarm for the FDR based distributed detection system is given by, 
\begin{equation}
P_{FA}=P(\Delta > T_{FDR};G_0)+\kappa P(\Delta =T_{FDR};G_0)
\end{equation}
where $\kappa$ is the randomization parameter. Similarly, the system-wide probability of detection is given by
\begin{equation}
P_D=P(\Delta > T_{FDR};G_1)+\kappa P(\Delta =T_{FDR};G_1)
\end{equation}
where the probability mass function (pmf) of the count statistic is given by Propositions 2 and 4 of \cite{ray_aes11}. %The system-wide performance metrics for the identical threshold scheme may be obtained similarly. 

It has been shown in \cite{ray_aes11} that determination of the optimal value of the FDR parameter $\gamma$, where optimality is with respect to system-level detection performance is a difficult problem. Optimizing the ROC, i.e., finding the optimal values of the parameter by maximizing the detection probability for a fixed false alarm probability, via simulation or numerical computation is very complex (see \cite{ray_aes11} for further discussion). A computationally less intensive approach is to use distance measure based optimization for system design. Motivated by this, in \cite{ray_aes11}, optimization of deflection coefficient of the count statistic ($\Delta$) was employed to find the optimal value of $\gamma$. Intuitively, an increased deflection coefficient generally implies greater separation between the pmfs of the count statistic under global hypotheses $G_0$ and $G_1$ and is likely to lead to better detector performance. Also, the asymptotic distribution of the count statistic is Gaussian for which it is known that maximization of the deflection coefficient leads to an optimal detector \cite{pincinbono95}. However, in this paper, we show that maximization of the deflection coefficient may not be the best design criterion for FDR based detection system under non-asymptotic conditions. Since the distribution of the count statistic ($\Delta$) is non-Gaussian in general, it is likely that the deflection coefficient fails to characterize its performance completely. We study the performance of several candidate distance measures (such as Kullback-Leibler Divergence, Bhattacharya Distance and Kolmogorov-Smirnov Distance) for system design.

As a motivating example, we perform the following simulation. Let $N=20$ sensors be randomly distributed within the circular ROI of radius $R=10$. The pdf of the sensor locations $r$ ($r$ is measured from the target location) is given by $f_R(r)=2r/R^2, 0\leq r\leq R$. The optimal value of parameter $\gamma$ (which maximizes the detection probability) obtained for a target model with power $P_0=5$ and radius of influence $d_0=5$ depends on the system-wide probability of false alarm, $P_{FA}$. For $P_{FA}=0.1$, the optimal parameter value that maximizes $P_D$ is found to be $\gamma_{opt}=0.25$. However, when deflection coefficient is used, we get the optimal parameter value as $\gamma_{opt}^d=0.008$. The optimization has been done numerically through simulations and the plots have been omitted for the sake of brevity of the paper.% from Figs. \ref{def_gamma_0} and \ref{def_qtau_0}.

%\begin{figure}[htb]
%\centering
%\includegraphics[width = 3.5in,height=!]{alpha_0Pd_gamma}\vspace{-0.75cm}
%\caption{Probability of detection versus FDR local parameter $\gamma$ when $P_{FA}=0.1$}
%\label{pd_gamma_0}
%\end{figure}
%
%\begin{figure}[htb]
%\centering
%\includegraphics[width = 3.5in,height=!]{alpha_0_Pd_qtau}\vspace{-0.75cm}
%\caption{Probability of detection versus identical local threshold parameter $p_{fa}=Q(\tau)$ when $P_{FA}=0.1$}
%\label{pd_qtau_0}
%\end{figure}

%\begin{figure}[htb]
%\centering
%\includegraphics[width = 3.5in,height=!]{alpha_0_def_gamma}\vspace{-0.75cm}
%\caption{Deflection Coefficient versus FDR local parameter $\gamma$}
%\label{def_gamma_0}
%\end{figure}
%
%\begin{figure}[htb]
%\centering
%\includegraphics[width = 3.5in,height=!]{alpha_0_def_qtau}\vspace{-0.75cm}
%\caption{Deflection Coefficient versus identical local threshold parameter $p_{fa}=Q(\tau)$}
%\label{def_qtau_0}
%\end{figure}

The above example shows that the deflection coefficient does not always yield the optimal value of the parameter and as shown can deviate substantially from the optimal value. This makes it necessary to come up with a distance measure which is computationally more efficient than using ROC based optimization, and provides a better approximation to the optimal solution. In this paper, we present some comparative performance results between four candidate measures: Deflection Coefficient \cite{Niu&Varshney:05eurasip}, Kullback-Leibler Divergence \cite{Lin_fusion06}, Bhattacharya Distance and Kolmogorov-Smirnov Distance \cite{rachev91}. Table II compares the optimal parameter values found by the optimization of each of the distance measures against the true optimal value found by ROC based optimization ($P_D$ maximization while fixing $P_{FA}=0.1$) for FDR based threshold design scheme.

\begin{table}[htb]
\caption{Threshold Comparison using different distance measures}
\begin{center}
\begin{tabular}{|p{2.25in}|p{2.0in}|}
\hline
	Metrics & FDR based threshold scheme ($\gamma$)\\
   \hline
   ROC based optimization  & 0.25\\
   %\hline 
   Deflection Coefficient optimization & 0.008\\ 
   %\hline
   Kullback-Leibler Divergence optimization & 0.15\\
   %\hline
   Bhattacharya Distance optimization & 0.2\\
   %\hline
   Kolmogorov-Smirnov Distance optimization & 0.2\\
   \hline
\end{tabular}
\end{center}
\label{FDR_table}
\end{table}

%\begin{center}
%\begin{table}[htb]
%\caption{Performance Comparison using different distance measures}
%\begin{tabular}{|p{2.25in}|p{2.0in}|p{2.25in}|}
%\hline
%	Metrics & FDR based threshold scheme ($\gamma$)& Identical threshold scheme ($Q(\tau)=p_{fa}$)\\
%   \hline
%   ROC based optimization  & 0.25 & 0.005\\
%   %\hline 
%   Deflection Coefficient optimization & 0.008 & 0.00005\\ 
%   %\hline
%   Kullback-Leibler Divergence optimization & 0.1 & 0.00005\\
%   %\hline
%   Bhattacharya Distance optimization & 0.2 & 0.01\\
%   %\hline
%   Kolmogorov-Smirnov Distance optimization & 0.2 & 0.005\\
%   \hline
%\end{tabular}
%\label{FDR_table}
%\end{table}
%\end{center}
%The K-S distance has the non-parametric property which suits our problem framework and thereby best approximates the optimal values of the parameters. 

For the example considered, as can be seen from Table \ref{FDR_table}, the Kolmogorov-Smirnov Distance (KSD) best approximates the optimal parameter value compared to other candidate measures for the FDR based scheme. We ran a large number of simulation experiments and similar results were obtained. The deflection coefficient fails to find the optimal parameter value as it only looks at the first and the second moments of the distributions which is not sufficient when the underlying distribution is non-Gaussian. In our context, the count statistic ($\Delta$) is non-Gaussian and is a discrete RV and, therefore, deflection coefficient does not characterize the distribution of count statistic. In the remainder of this section, we formally define the Kolmogorov-Smirnov distance and perform some empirical studies based on K-S distance to find the optimal parameter value for the FDR based distributed detection scheme.

\begin{definition} The Kolmogorov-Smirnov distance (K-S distance) is defined as the maximum value of the absolute difference between two cumulative distribution functions \cite{rachev91}.
\end{definition}

The Kolmogorov Smirnov distance measures the distance between the empirical distribution function of the sample data and the cumulative distribution function of the reference distribution. The K-S distance is distribution free, i.e., it makes no assumptions about the underlying data distributions. The K-S distance between two cdfs $F(\cdot)$ and $G(\cdot)$ is defined as

\begin{equation}
KSD(F(\cdot),G(\cdot))=\sup_{0\leq x\leq 1}|F(x)-G(x)|
\end{equation}

In our context, $F(\cdot)$ and $G(\cdot)$ represent the cdfs of the count statistic ($\Delta$) under global hypotheses $G_0$ and $G_1$. The optimal parameter value of the local decision threshold parameter ($\gamma$ for FDR based scheme) is found by maximizing the K-S distance. %Figs. \ref{ksd_gamma} and \ref{ksd_qtau} show the optimal parameter values found by K-S distance maximization which is close to the true value found by using $P_D$ maximization.

%\begin{figure}[htb]
%\centering
%\includegraphics[width = 3.5in,height=!]{ksd_P0_5_d0_5}\vspace{-0.75cm}
%\caption{Kolmogorov-Smirnov distance versus FDR local parameter $\gamma$ for $P_{FA}$=0.1}
%\label{ksd_gamma}
%\end{figure}
%
%\begin{figure}[htb]
%\centering
%\includegraphics[width = 3.5in,height=!]{ksd_qtau_P0_5_d0_5}\vspace{-0.75cm}
%\caption{Kolmogorov-Smirnov distance versus identical local threshold parameter $p_{fa}$}
%\label{ksd_qtau}
%\end{figure}

It is important to note here that the optimal parameter value depends on the system requirement of system-wide probability of false alarm ($P_{FA}$). When $P_{FA}=0.2$ is used, the optimal parameter value based on ROC optimization comes out to be $\gamma=0.2$ which again is close to the optimal parameter value found by using the K-S distance measure for optimization. %It is important to note here that although the optimal parameter value changes, the FDR based decision threshold scheme proposed in \cite{ray_aes11} outperforms the identical threshold scheme even when K-S distance is used as the distance metric as seen in Fig. \ref{ROC_KSD_optimal}. 
Fig. \ref{ROC_KSD_deff} shows the ROCs obtained by using the parameters obtained by K-S distance optimization and deflection coefficient optimization for the FDR based scheme. Note that in Fig. \ref{ROC_KSD_deff}, the ROC obtained by deflection coefficient optimization is not as good as the ROC obtained by using K-S distance optimization.

%\begin{figure}[htb]
%\centering
%\includegraphics[width = 3.5in,height=!]{ROC_KSD_optimal}
%\caption{ROC using the optimal parameter values found using K-S distance}
%\label{ROC_KSD_optimal}
%\end{figure}

\begin{figure}[htb]
\centering
\includegraphics[width = 3.5in,height=!]{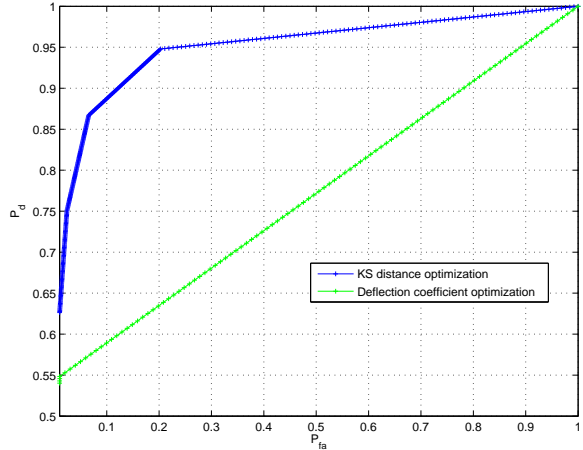}
\caption{ROC using the optimal parameter values found using K-S distance and deflection coefficient}
\label{ROC_KSD_deff}
\end{figure}

%\begin{figure}[htb]
%\centering
%\includegraphics[width = 3.5in,height=!]{ROC_deff_optimal}
%\caption{ROC using the optimal parameter values found using deflection coefficient}
%\label{ROC_deff_optimal}
%\end{figure}
%\textbf{I think we dont need figure ROCdeffoptimal anymore}

From empirical studies, K-S distance optimization seems to provide a good approximation to the optimal value. We hypothesize that since K-S distance measures the overlap between the distributions and is distribution free in nature, it provides better results compared to deflection coefficient maximization. However, to find the best distance measure to use is an interesting but difficult problem to solve considering the large number of candidate measures available in the literature (Ali-Silvey Distance Measures \cite{basseville89}). We leave this for future work and for the remaining part of the paper, we use K-S distance or ROC based optimization to find the optimal value of the local threshold parameter.

\section{Control of FDR in the presence of Byzantines}
\label{Byz}
In this section, we consider the problem of FDR based distributed detection in the presence of Byzantines. We first show the effect of Byzantines on the control of FDR in multiple comparison problems and  in subsequent sections analyze the effect of Byzantines on FDR based distributed detection.

Byzantines are the local sensors which send falsified information to the FC to deteriorate system performance. Since control of FDR is based on p-values, the Byzantines' attack strategy would be to report the decision based on a transformed p-value denoted by $q_i=h(p_i)$, where $h(\cdot)$ is a transformation used by the Byzantines and $p_i$ is the true p-value of the $i^{th}$ sensor. The transformation $h(\cdot)$ needs to satisfy the properties listed below:
\bi
\item[1.] $h(\cdot)$ should be a function whose domain and range are $[0,1]$, i.e., $h:[0,1]\to[0,1]$
\item[2.] $h(\cdot)$ should be a decreasing function. This property is essential since the Byzantines' aim is to deteriorate detection performance \cite{marano_tsp09}. As the p-value represents the `confidence' of the target being absent, they would like to report falsified information by reversing it. A decreasing function ensures a higher q-value for a lower p-value and vice-versa.
\ei

One possible transformation is the linearly decreasing function $h(p)=1-p$, which is analogous to flipping of the local decision. The transformation corresponding to flipping has been shown to be the optimal attack for Byzantines in distributed detection using identical local decision rules \cite{rawat_tsp11} and target localization with quantized data \cite{Vempaty_TSP}. Also, as shown later in this section, the above transformation ensures that the Byzantines attack the network in a covert manner. This attack strategy reduces the detection performance of the network (refer to Fig. \ref{det_H01}) while controlling the FDR value at the pre-determined threshold $\gamma$ (refer to Fig. \ref{FDR_H01}). Hence, the attack strategy $h(p)=1-p$ allows the Byzantines to reduce the network performance while not changing their behavior in a manner that can be detected by observers such as the FC.

In the remainder of this paper, we use the above transformation to model the Byzantine attack. Therefore, if $p_i$ represents the true p-value of the $i^{th}$ sensor, then the transformed p-values are given by
\begin{equation}
\label{byz_model}
q_i=
\begin{cases}
p_i,	& \text{if $i^{th}$ sensor is honest}\\
1-p_i,	& \text{if $i^{th}$ sensor is Byzantine}
\end{cases}
\end{equation}

In the rest of the section, we show the effect of Byzantines on the control of FDR. As mentioned previously, an important aspect of FDR based distributed detection is the control of FDR value at the pre-determined threshold $\gamma$. The FDR control algorithm provided earlier in Section \ref{FDR}, controls the FDR value at $\gamma$ when the true Hypothesis is $H_0$ and at a value less than or equal to $\gamma$ when the true Hypothesis is $H_1$. We now prove that the FDR value is controlled even in the presence of Byzantines.

%\textbf{discontinuity here}
\begin{proposition}
Let $N$ sensors be randomly deployed in the ROI. The local sensors report local decisions to the FC based on their p-values and the FC makes the final decision regarding the presence/absence of the target using these local decisions. Let there be $M$ Byzantines in the network which transform the p-values according to \eqref{byz_model} and report decisions based on the transformed p-values. For independent local sensors under global hypothesis $G_0$, the FDR control algorithm proposed in \cite{benj95} controls the FDR value at the pre-determined threshold $\gamma$, even in the presence of Byzantines. %\textbf{lets talk}
\end{proposition}
\begin{proof}
%\textbf{cite equations by %\eqref}
In order to prove this proposition, we need to find the distribution of p-values under $H_0$ in the presence of Byzantines. The Byzantines' model used in this paper is that of transforming the p-values according to \eqref{byz_model}. %For a Gaussian random variable 
%$N(\phi,1)$, the density of the p-value $f_\phi(u)$ is given by \cite{ray_aes11}
%\begin{equation}
%\label{fu}
%f_\phi(u)=\exp\left(-\frac{\phi^2}{2}\right)\exp(\phi Q^{-1}(u)), \qquad \text{$0 \leq u \leq 1$}
%\end{equation}
When the Byzantines are absent, the true p-value represents the cumulative distribution function under $H_0$ and, therefore, follows uniform distribution. Since the Byzantines' effect is a transformation given by \eqref{byz_model}, the distribution of the falsified p-values under $H_0$ can be found by the transformation of random variables $V=1-U$. Due to probability integral transform, the falsified p-values also follow uniform distribution under null hypothesis.

%The true p-value under the alternate hypothesis $H_1$ follows the distribution given by \eqref{fu} with $\phi = \sqrt{P_0}$. For the Byzantines, the falsified p-value is given by $q=1-p$. Using the change of variables, $V=1-U$, 
%\begin{equation}
%f_\phi(v)=f_\phi(u)|_{u=1-v}\left|\frac{du}{dv}\right|
%\end{equation}
%This gives us the distribution of the falsified p-values $f_V(v)$ as
%\begin{equation}
%\label{fv}
%f_V^B(v)=\exp\left(-\frac{\phi^2}{2}\right)\exp(\phi Q^{-1}(1-v))	\qquad \text{for 0 $\leq v \leq$ 1}
%\end{equation}
%where $\phi$ is the received signal amplitude at the local sensor which is either 0 or $\sqrt{P_0}$ under hypothesis $H_0$ and $H_1$ respectively.

The proof of this proposition then follows from the straightforward observation that the FDR algorithm proposed by Benjamini and Hochberg \cite{benj95} controls the FDR value for any configuration (distribution) of the false null hypothesis. The condition that p-values are uniformly distributed under true null hypothesis is still satisfied and since the proof does not depend on the p-values' distribution under $H_1$, the FDR value is controlled at the pre-determined threshold.
\end{proof}

%\begin{proposition}
%\label{byz_dist}
%The distribution of reported p-values is uniform in the absence of target and in the presence of target, the reported p-values of Byzantine sensors follow the distribution $f^B_V(v)$ given by
%
%\end{proposition}
%
%\begin{proof}
%\textbf{the proof is not well-written. We need to show that the Byzantine distribution under $H_0$ is uniform (due to perhaps probability integral transform). We start from honest sensors distribution. Show that it is uniform. Then make transformation U = 1- V and justify it is uniform. I mean the proof in the present form is not very clear)}
%
%\end{proof}

We now validate via simulations that the FDR value is controlled at $\gamma$ even in the presence of Byzantines. Fig. \ref{FDR_H01} shows the FDR value under both the hypotheses $G_0$ and $G_1$ with varying fraction of Byzantines ($\alpha$) present in the network. Let us consider a distributed detection system with the parameters: $N=20$, $R=10$, $d_0=5$, $P_0=5$ and the FDR parameter $\gamma=0.25$. The simulation results are for $5 \times 10^4$ Monte-Carlo runs. As can be seen from the figure, the FDR value is maintained at $\gamma$ under $G_0$ and at a value less than or equal to $\gamma$ under $G_1$ even in the presence of Byzantines.

%\begin{figure}[htb]
%\centering
%\includegraphics[width = 3.5in,height=!]{FDR_H0}
%\caption{FDR value against the fraction of Byzantines when the true hypothesis is $G_0$}
%\label{FDR_H0}
%\end{figure}
%
%\begin{figure}[htb]
%\centering
%\includegraphics[width = 3.5in,height=!]{FDR_H1}
%\caption{FDR value against the fraction of Byzantines when the true hypothesis is $G_1$}
%\label{FDR_H1}
%\end{figure}

%\begin{figure*}[!t]
%            	\centerline
%            	{
%                	\subfloat[FDR value against the fraction of Byzantines when the true hypothesis is $G_0$]{\includegraphics[width=3in]{FDR_H0}%
%                	\label{FDR_H0}}
%                	\hfil
%                	\subfloat[FDR value against the fraction of Byzantines when the true hypothesis is $G_1$]{\includegraphics[width=3in]{FDR_H1}%
%                	\label{FDR_H1}}
%            	}
%            	\caption{FDR value against the fraction of Byzantines}
%            	\label{FDR_Byz}
%        	\end{figure*}

\begin{figure}[htb]
\centering
\includegraphics[width = 3.5in,height=!]{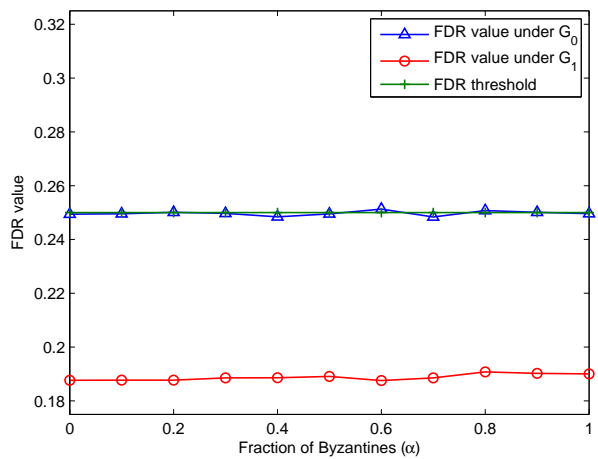}
\caption{FDR value against the fraction of Byzantines}
\label{FDR_H01}
\end{figure}

%\textbf{below, I think where you use H0/H1 it should be G0/G1....}

This ineffectiveness of Byzantines on the FDR value is expected because the effect of Byzantines changes the order of the transformed p-values (q-values) and the largest threshold crossing index would be different as compared to the largest threshold crossing index for the original p-values. In other words, in the presence of a target, most of the true p-values are small and therefore the threshold crossing would be closer to the right extremal resulting in a high number of detections as depicted by an example in Fig. \ref{threshold_true}. In the presence of Byzantines, the p-values get transformed as defined by \eqref{byz_model} and the transformed p-values (q-values) become larger and the threshold crossing shifts to the left, as shown in Fig. \ref{threshold_false}. This reduces the number of detections and it is equivalent to looking at an earlier threshold crossing index for the true p-values. As pointed out in Section \ref{FDR}, the FDR algorithm looks at the largest index satisfying $p_{(i)} \leq i\gamma/N$ to maximize the power of the test. Observe that when $\alpha=1$, i.e., all the sensors are Byzantines, the order of the q-values is reversed as compared to the p-values and the FDR algorithm based on the q-values ends up looking at the smallest index of p-values satisfying $p_{(i)} \leq i\gamma/N$ rather than the largest index. Under this observation, we conjecture that for $0 < \alpha < 1$, the FDR algorithm based on the q-values ends up looking at an index of p-value between the largest and the smallest indices satisfying $p_{(i)} \leq i\gamma/N$. Since the proof of control of FDR does not depend on whether it is the largest threshold crossing index or not, the FDR value is maintained at the required threshold. However, the power of the test degrades in the presence of Byzantines. 

%\begin{figure}[htb]
%\centering
%\includegraphics[width = 3.5in,height=!]{threshold_true}\vspace{-0.75cm}
%\caption{True p-values against linearly increasing threshold in the presence of target showing large number of detections ($\Delta$)}
%\label{threshold_true}
%\end{figure}
%
%\begin{figure}[htb]
%\centering
%\includegraphics[width = 3.5in,height=!]{threshold_false}\vspace{-0.75cm}
%\caption{Reported p-values in the presence of Byzantines against linearly increasing threshold in the presence of target showing reduction in $\Delta$, number of detections}
%\label{threshold_false}
%\end{figure}

\begin{figure*}[!t]
            	\centerline
            	{
                	\subfloat[True p-values resulting in large number of detections ($\Delta$)]{\includegraphics[width=3in]{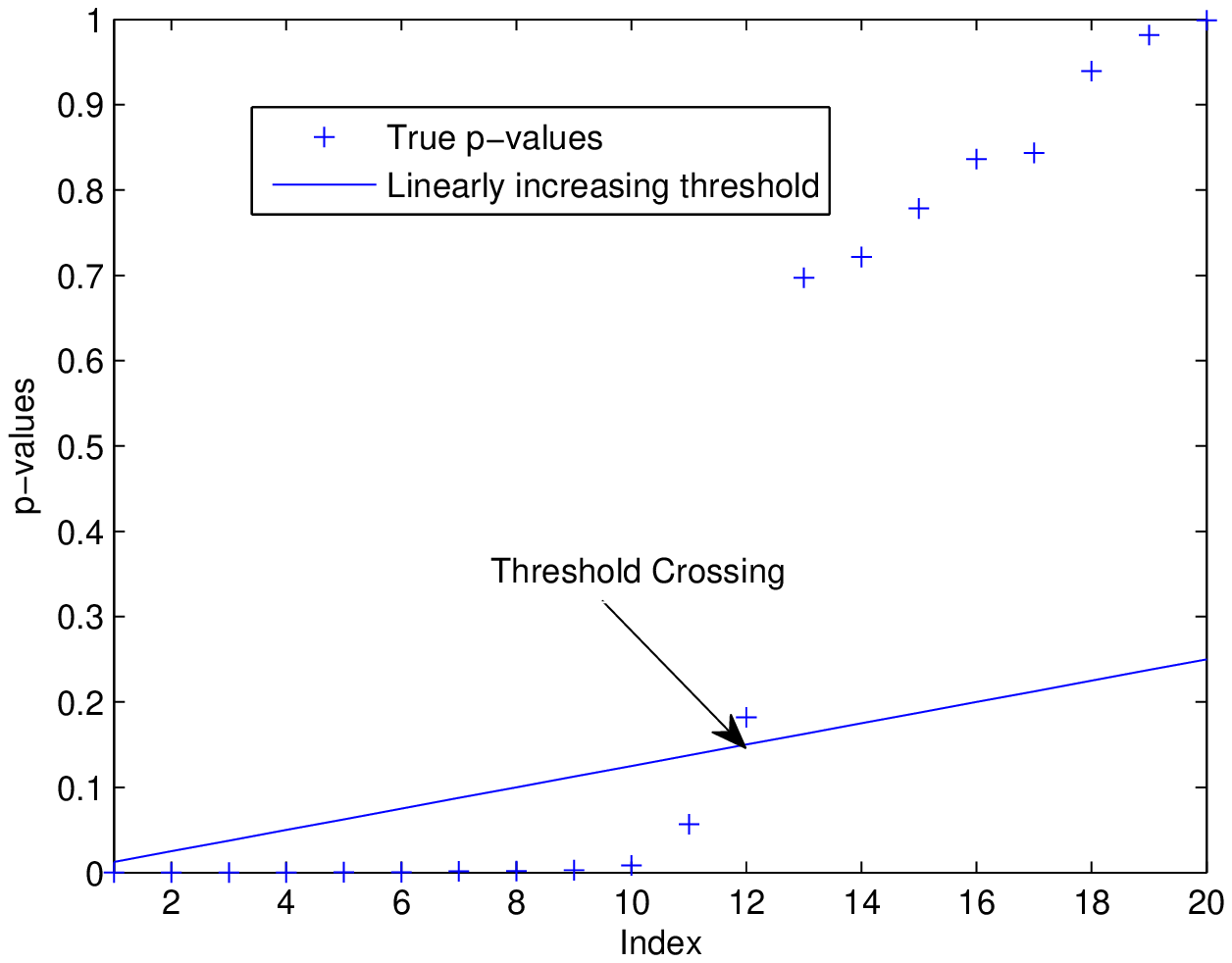}%
                	\label{threshold_true}}
                	\hfil
                	\subfloat[Falsified p-values in the presence of Byzantines resulting in reduction in $\Delta$]{\includegraphics[width=3in]{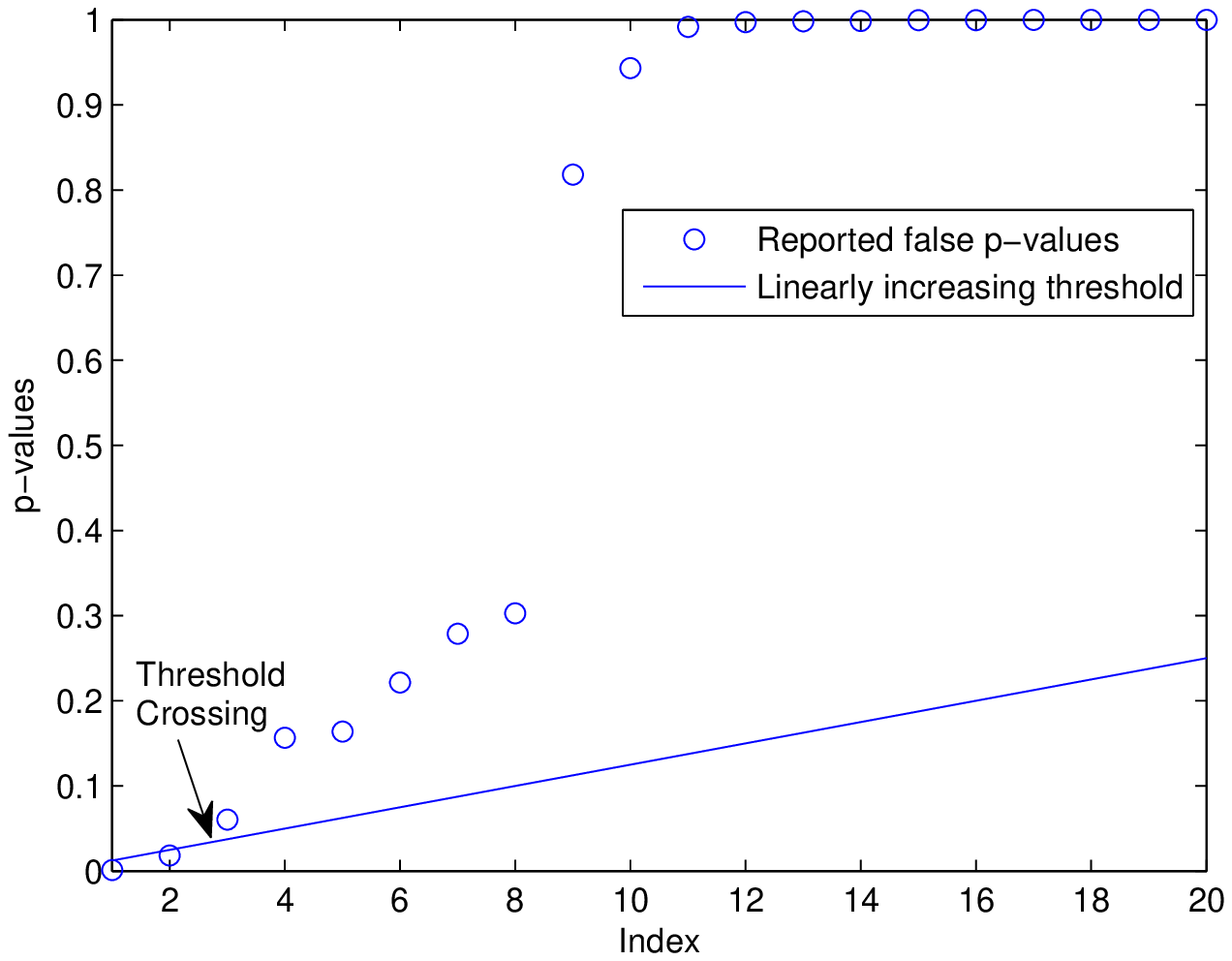}%
                	\label{threshold_false}}
            	}
            	\caption{p-values against linearly increasing threshold in the presence of target}
            	\label{p-values}
        	\end{figure*}

%\textbf{need to put these figures side by side and explain better}

In the above discussion, we have shown that the FDR value is not affected by the presence of Byzantines. However, it was also pointed out that the FDR control algorithm for the transformed p-values (q-values) is equivalent to looking at an earlier index on the true p-values. In the presence of Byzantines, the number of detections is reduced and, therefore, the distribution of count (number of detections) under $G_1$ is now closer to the distribution of the count statistic under $G_0$ which remains unchanged in the presence of Byzantines. This makes it difficult to distinguish between the two hypotheses. Fig. \ref{det_H01} shows the effect of Byzantines on the average number of detections reported by the local sensors. The simulation parameters are the same as before: $N=20$, $R=10$, $d_0=5$, $P_0=5$ and the FDR parameter $\gamma=0.25$. The plots are for $5 \times 10^4$ Monte-Carlo runs and it can be observed that, in the presence of Byzantines, the number of detections remains the same under $G_0$ but reduces under $G_1$. When $\alpha$ approaches `1', all the sensors are Byzantines and their local hypotheses are effectively reversed. Therefore, at $\alpha=1$, the average number of detections is lower under $G_1$ than the value under $G_0$.

\begin{figure}[htb]
\centering
\includegraphics[width = 3.5in,height=!]{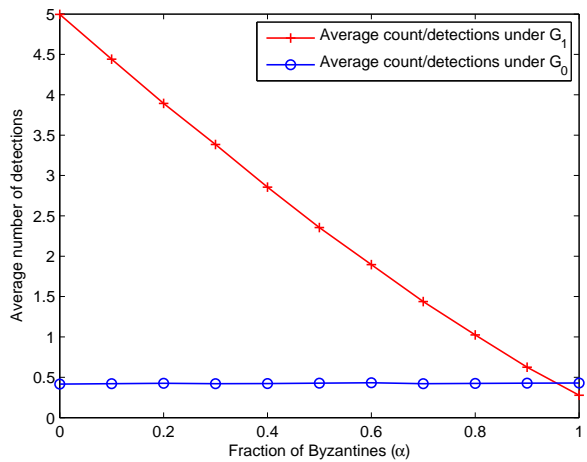}
\caption{Average number of detections against the fraction of Byzantines}
\label{det_H01}
\end{figure}

In the following section, we explore this intuitive observation that the Byzantines bring the distribution of the count statistic under $G_1$ closer to its distribution under $G_0$ in the context of FDR based distributed detection. We show analytically the effect of Byzantines on the count statistic and derive the distributions of the count statistic under $G_0$ and $G_1$.
%\textbf{this is good, but we need to make it a little concise....we can talk}

\section{FDR based Distributed detection in the presence of Byzantines}
\label{perf}
In order to understand the behavior of FDR based distributed detection scheme in the presence of Byzantines, we require the knowledge regarding the pmf of the count statistic ($\Delta$) in the presence of Byzantines. These results are provided in the rest of this section.

\subsection{Probability mass function of the Count Statistic ($\Delta$)}
\label{pmf}
Let the observed p-values be denoted by the random variables $\left\{U_i\right\}_{i=1,\ldots,N}$ and the transformed p-values (q-values) be denoted by the transformed random variables $\left\{V_i\right\}_{i=1,\ldots,N}$ where the transformation is as follows
\begin{equation}
V_i=
\begin{cases}
U_i	&	\text{if $i$ is an honest sensor}\\
1-U_i	&	\text{if $i$ is a Byzantine sensor}
\end{cases}
\end{equation}

%\subsubsection*{FDR based threshold design}

\begin{proposition}
\label{prop_FDR_G0}
The probability of $\Delta$= i local false alarms (count) for $N$ sensors containing $M=\alpha N$ Byzantines in the ROI, and control of FDR at $\gamma$ under $G_0$ (absence of target in the ROI) is given by
\begin{equation}
P(\Delta =i;G_0)=\binom{N}{i}(1-\gamma)\Big(\frac{i\gamma}{N}\Big)^i\Big(1-\frac{i\gamma}{N}\Big)^{N-i-1}	
\end{equation}
\end{proposition}

In the absence of a target, the p-values of both the Byzantines and the honest nodes are uniformly distributed, that is both $U_i$ and $V_i$ are uniformly distributed. Therefore, the result remains the same as derived by Finner and Roters \cite{finner02} using Dempster's formula for barrier crossing for uniform random variables, irrespective of the presence of Byzantines. Similarly, the asymptotic distribution can be found as

\begin{equation}
\lim_{N \to \infty}{P(\Delta =i; G_0)}=\frac{i^i}{i!}(1-\gamma)\gamma^i \exp{(-i\gamma)}
\end{equation}

\begin{proposition}
\label{prop_FDR_G1}
The probability of $\Delta =i$ local detections (count) for $N$ sensors containing $M=\alpha N$ Byzantine sensors in the ROI, and control of FDR at $\gamma$ under $G_1$ (presence of target in ROI) is given by

\begin{multline}
\label{Eq:FDR_G1}
P(\Delta = i; G_1)=\frac{1}{\binom{N}{M}} \sum\int_{v_{N,N}=\gamma}^1 \cdots \int_{v_{i+1,N}=(i+1)\gamma/N}^{v_{i+2,N}}\int_{v_{i,N}=0}^{i\gamma /N}\\
\cdots \int_{v_{1,N}=0}^{v_{2,N}}N!f_{V_1}(v_1)f_{V_2}(v_2)\cdots f_{V_N}(v_N)dv_{1,N}dv_{2,N}\cdots dv_{N,N}
\end{multline}
Also, asymptotically, i.e., for large N,
\begin{equation}
P(\Delta = i; G_1)= \sum_{k=\max{(0,M-N+i)}}^{\min{(M,i)}}\binom{M}{k}\binom{N-M}{i-k}(\bar{p}_D^B)^k(1-\bar{p}_D^B)^{M-k}(\bar{p}_D^H)^{i-k}(1-\bar{p}_D^H)^{(N-M)-(i-k)}
\end{equation}
or 
\begin{equation}
P(\Delta = i; G_1)= \binom{N}{i}(\bar{p}_D)^i(1-\bar{p}_D)^{N-i}
\end{equation}
where $\bar{p}_D^H$ is the average probability of reporting `1' under $G_1$ for an honest local sensor, $\bar{p}_D^B$ is the average probability of reporting `1' under $G_1$ for a Byzantine local sensor and $\bar{p}_D=\alpha\bar{p}_D^B +(1-\alpha)\bar{p}_D^H $ is the average probability of reporting `1' under $G_1$.
\end{proposition}

\begin{proof}
The proof is provided in Appendix \ref{proof_FDR_G1}

\end{proof}

It is interesting to observe here that the Byzantines only affect the pmf of the count statistic ($\Delta$) under $G_1$, while the pmf under $G_0$ remains the same. The reason behind this is that the p-values under $G_0$ are uniformly distributed and a transformation $h(p)=1-p=q$ still keeps the q-values uniformly distributed under $G_0$. However, the pmf under $G_1$ changes as the p-values are no longer uniformly distributed.  %\textbf{we will come back to this later}

\subsection{Numerical Results}
In this section, we provide numerical and simulation results to validate the analytical expressions obtained for the pmf of the count statistic under FDR based threshold design. Tables \ref{tab:FDRal0}, \ref{tab:alFDR0.5} and \ref{tab:alFDR1} show the numerical and simulation results of analytically derived $P(\Delta = i; G_1)$ for the FDR based scheme for different values of $\alpha$. The signal and target parameters are: $P_0=3$, $R=10$, $d_0=3$, $N=4$ and FDR parameter $\gamma=0.1$. The integrals given in \eqref{Eq:FDR_G1} have been evaluated using Monte Carlo integration methods \cite{press92}. The simulation results are for $5 \times 10^5$ Monte Carlo runs. These tables show that the numerical and simulation results match very closely. Also observe that the probability that count is `0', $P(\Delta=0;G_1)$, increases with $\alpha$ while $P(\Delta=i;G_1)$ for $i\neq 0$ decreases with $\alpha$, suggesting a shift of the distribution towards `0'.

\begin{table}[htb]
\caption{Numerical and Simulation results for $P(\Delta =i;G_1)$ for control of FDR when $\alpha=0$} \centering
\begin{tabular}{|c| c c c c c|}
\hline
   Count & 0  & 1 & 2 & 3 & 4\\
   \hline 
   $P(\Delta;G_1)$ (Numerical) & 0.7777 & 0.1777 & 0.0407 & 0.0061 & 3.1544 x $10^{-4}$\\ 
   \hline
   $P(\Delta;G_1)$ (Simulations)& 0.7756 & 0.1773 & 0.0401 & 0.0065 & 5.1000 x $10^{-4}$\\
   \hline
\end{tabular}
\label{tab:FDRal0}
\end{table}

\begin{table}[h]
\caption{Numerical and Simulation results for $P(\Delta =i;G_1)$ for control of FDR when $\alpha=0.5$} \centering
\begin{tabular}{|c| c c c c c|}
\hline
   Count & 0  & 1 & 2 & 3 & 4\\
   \hline 
   $P(\Delta;G_1)$ (Numerical) & 0.8355 & 0.1347 & 0.0227 & 0.0027 & 1.8121 x $10^{-4}$\\ 
   \hline
   $P(\Delta;G_1)$ (Simulations)& 0.8383 & 0.1353 & 0.0233 & 0.0028 & 1.7000 x $10^{-4}$\\
   \hline
\end{tabular}
\label{tab:alFDR0.5}
\end{table}

\begin{table}[h]
\caption{Numerical and Simulation results for $P(\Delta =i;G_1)$ for control of FDR when $\alpha=1$} \centering
\begin{tabular}{|c| c c c c c|}
\hline
   Count & 0  & 1 & 2 & 3 & 4\\
   \hline 
   $P(\Delta;G_1)$ (Numerical) & 0.9022 & 0.0789 & 0.0106 & 0.0012 & 6.2896 x $10^{-5}$\\ 
   \hline
   $P(\Delta;G_1)$ (Simulations)& 0.9084 & 0.0793 & 0.0111 & 0.0011 & 6.4000 x $10^{-5}$\\
   \hline
\end{tabular}
\label{tab:alFDR1}
\end{table}

In Fig. \ref{asymp_all}, we provide the simulation results and the analytical approximation for the pmf of the count statistic $P(\Delta;G_1)$ for a large number of sensors in the ROI. The simulation parameters are $N=500$, $P_0=15$, $R=10$, $d_0=3$ and FDR parameter $\gamma=0.0077$. The simulation results are for $5 \times 10^4$ Monte Carlo runs. The simulation results demonstrate that the asymptotic expressions for $P(\Delta;G_1)$ match the simulation results very well. Again note that the distributions move closer to `0' as $\alpha$ increases suggesting a decrease in the number of detections.

\begin{figure}[htb]
\centering
\includegraphics[width = 3.5in,height=!]{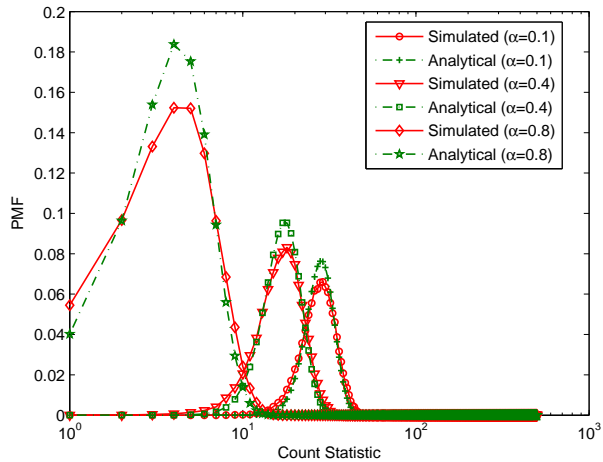}
\caption{Simulated and analytical results for $P(\Delta;G_1)$ for the FDR based scheme under asymptotic conditions when $\alpha=0.1,0.4,0.8$}
\label{asymp_all}
\end{figure}

%\begin{figure}[htb]
%\centering
%\includegraphics[width = 3.5in,height=!]{asymp_010_mixture}
%\caption{Simulated and analytical results for $P(\Delta;G_1)$ for the FDR based scheme under asymptotic conditions when $\alpha=0.1$}
%\label{asymp_010}
%\end{figure}
%
%\begin{figure}[htb]
%\centering
%\includegraphics[width = 3.5in,height=!]{asymp_040_mixture}
%\caption{Simulated and analytical results for $P(\Delta;G_1)$ for the FDR based scheme under asymptotic conditions when $\alpha=0.4$}
%\label{asymp_040}
%\end{figure}
%
%\begin{figure}[htb]
%\centering
%\includegraphics[width = 3.5in,height=!]{asymp_080_mixture}
%\caption{Simulated and analytical results for $P(\Delta;G_1)$ for the FDR based scheme under asymptotic conditions when $\alpha=0.8$}
%\label{asymp_080}
%\end{figure}

Fig. \ref{Pd_alpha_pf_01_adap} shows the reduction in detection performance with an increase in the number of Byzantines in the network. The simulation parameters are: $R=10$, $d_0=5$, $P_0=5$, $N=20$, system-wide probability of false alarm is fixed at $P_{FA}=0.1$. This yields optimal FDR parameter as $\gamma=0.25$. The simulation results are for $1 \times 10^4$ Monte-Carlo runs. This shows that the Byzantines reduce the power of the test (detection probability) even though the FDR value is maintained at the prescribed threshold. This leads to the interesting question regarding the optimality of the current transformation $h(\cdot)$ for the Byzantines. This analysis is much more complex and will be explored in the future.

We also compared the detection performance of FDR based threshold scheme and identical threshold based scheme in the presence of Byzantines. In our previous work \cite{ray_aes11}, we have shown that the FDR based threshold scheme performs better than the identical threshold scheme in the absence of Byzantines. Identical threshold scenario under Byzantine attack has been investigated by us \cite{rawat_tsp11}. The identical threshold on the observations is represented as $\tau$ or, equivalently, the threshold on the p-values is $p_{fa}=Q(\tau)$. We skip the details of the optimal threshold ($\tau$ or $p_{fa}$) design since it is not the focus of this paper\footnote{Details of optimal design under identical threshold design in the presence of Byzantines will be published as a technical report.}. However, we present a numerical comparison of the two schemes. Fig.~\ref{ROC_04} presents a comparison of the two schemes when the fraction of Byzantines, $\alpha=0.4$. The simulation parameters are: $N=20$, $R=10$, $d_0=5$, $P_0=5$, the optimal FDR parameter $\gamma=0.1$ and the optimal identical threshold parameter $Q(\tau)=p_{fa}=0.005$.

\begin{figure}[htb]
\centering
\includegraphics[width = 3.5in,height=!]{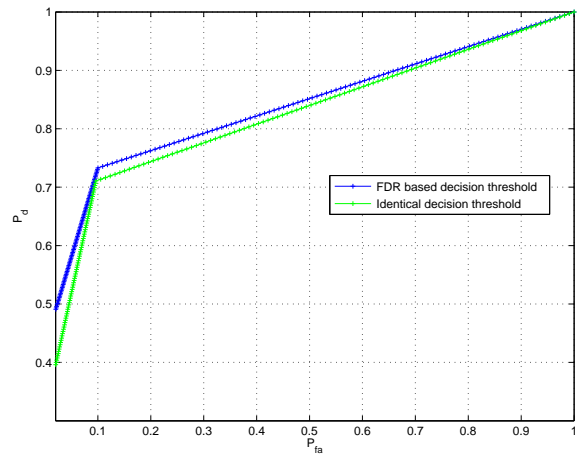}
\caption{ROC using the optimal parameter values found using K-S distance when $\alpha=0.4$}
\label{ROC_04}
\end{figure}

\section{Adaptive FDR based distributed detection}
\label{adap}
In this section, we first demonstrate that the optimal value of the design parameter for FDR based scheme depends on the fraction of Byzantines ($\alpha$). If prior information is available regarding $\alpha$, we can design our system such that the performance is optimized for the given value of $\alpha$. However, in a dynamically changing environment, it becomes important that we learn this fraction ($\alpha$) over time and change our system design parameter values adaptively. In this section, we propose an adaptive algorithm which learns the maliciousness of network over time and changes the values of design parameters to improve the system's detection performance in a dynamic manner. In other words, we learn the effect of Byzantines on the network and reduce their effect by adaptively changing the values of system parameters.

\subsection{Design of optimal parameter value}
\label{opt}
In this subsection, we first give design guidelines for FDR based distributed detection in the presence of Byzantines. For a fixed system-wide probability of false alarm given by
\begin{equation}
P_{FA}=P(\Delta > T;G_0)+\kappa P(\Delta =T;G_0)
\end{equation}
where $T$ is the global threshold for the count statistic used at the FC and $\kappa$ is the randomization parameter, the optimal value of local threshold parameter $\gamma$ is found by maximizing the system-wide probability of detection ($P_D$). $P_D$ is given by
\begin{equation}
P_D=P(\Delta > T;G_1)+\kappa P(\Delta =T;G_1)
\end{equation}
where the pmf of the count statistic is given by Propositions \ref{prop_FDR_G0} and \ref{prop_FDR_G1}. As can be seen from Proposition \ref{prop_FDR_G0}, the distribution of the count statistic under $G_0$ does not depend on $\alpha$ and, therefore, the expression of $P_{FA}$ remains the same irrespective of the value of $\alpha$. In this section, we show through simulations that the optimal parameter value for FDR based approach varies with $\alpha$. Intuitively, the Byzantines decrease the distance (such as KL divergence or KS distance) between the pmfs of the count statistic under global hypotheses $G_0$ and $G_1$. It is, therefore, important to re-optimize the values of local threshold parameters to increase the distance between these pmfs as much as possible for this fixed value of $\alpha$. 
% In the presence of Byzantines, the optimal parameter value for Identical threshold scheme remains the same for varying values of $\alpha$ as seen in Fig. \ref{pd_qtau}. 
%\textbf{we need to be consistent ..FDR based or FDR-based}

For system and target parameters given by $N=20$, $R=10$, $P_0=5$, $d_0=5$ and $P_{FA}=0.1$, we have shown in Section \ref{DM} that the optimal value of local threshold parameter is $\gamma=0.25$ for the FDR based scheme. For the FDR based algorithm, the optimal values of $\gamma$ changes with $\alpha$ from $\gamma_{low}=0.25$ to $\gamma_{high}=0.1$ as shown in Fig. \ref{pd_gamma}. The simulations were performed for several values of $\alpha$ ranging between 0 and 1 but a number of curves have been omitted for the sake of brevity.

%\begin{figure}[htb]
%\centering
%\includegraphics[width = 3.5in,height=!]{alpha_Pd_qtau}\vspace{-0.75cm}
%\caption{Probability of detection versus identical local parameter $Q(\tau)$ when $P_{FA}=0.1$ for varying $\alpha$}
%\label{pd_qtau}
%\end{figure}

%\begin{figure}[htb]
%\centering
%\includegraphics[width = 3.5in,height=!]{alpha_04_Pd_qtau}\vspace{-0.75cm}
%\caption{Probability of detection versus identical local parameter $Q(\tau)$ when $P_{FA}=0.1$ and $\alpha=0.4$}
%\label{pd_qtau_04}
%\end{figure}

%\begin{figure}[htb]
%\centering
%\includegraphics[width = 3.5in,height=!]{alpha_08_Pd_qtau}\vspace{-0.75cm}
%\caption{Probability of detection versus identical local parameter $Q(\tau)$ when $P_{FA}=0.1$ and $\alpha=0.8$}
%\label{pd_qtau_08}
%\end{figure}

\begin{figure}[htb]
\centering
\includegraphics[width = 3.5in,height=!]{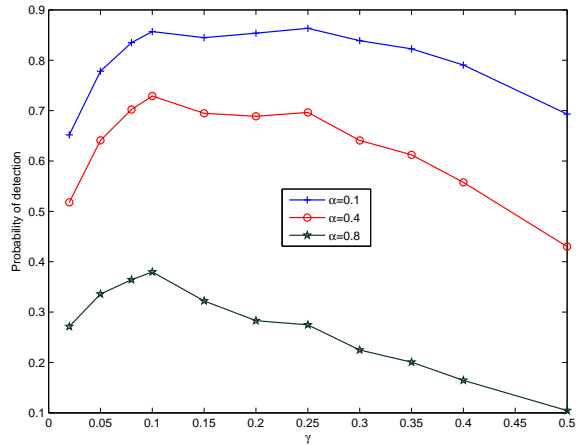}\vspace{-0.75cm}
\caption{Probability of detection versus FDR local parameter $\gamma$ when $P_{FA}=0.1$ and for varying $\alpha$}
\label{pd_gamma}
\end{figure}

%\begin{figure}[htb]
%\centering
%\includegraphics[width = 3.5in,height=!]{alpha_04_Pd_gamma}\vspace{-0.75cm}
%\caption{Probability of detection versus FDR local parameter $\gamma$ when $P_{FA}=0.1$ and $\alpha=0.4$}
%\label{pd_gamma_04}
%\end{figure}

%\begin{figure}[htb]
%\centering
%\includegraphics[width = 3.5in,height=!]{alpha_08_Pd_gamma}\vspace{-0.75cm}
%\caption{Probability of detection versus FDR local parameter $\gamma$ when $P_{FA}=0.1$ and $\alpha=0.8$}
%\label{pd_gamma_08}
%\end{figure}

From extensive simulations, it was observed that the optimal parameter value remains nearly constant for different intervals of $\alpha$. In the particular example chosen, for $\alpha \leq  0.2$, $\gamma_{opt}=\gamma_{low}=0.25$ and for $\alpha > 0.2$, $\gamma_{opt}=\gamma_{high}=0.1$. Using the optimal parameter values, we determine the probability of detection as a function of $\alpha$ via simulation. Fig. \ref{Pd_alpha_pf_01_adap} shows the improvement in detection performance of the FDR based scheme using adaptive optimal parameter values. From Fig. \ref{Pd_alpha_pf_01_adap}, we can also observe that the FDR based scheme using fixed parameter values fails under severe Byzantine attack ($\alpha =1$) since it performs worse than a random decision making mechanism which gives $P_D=P_{FA}$. However, under severe attack, the adaptive FDR based scheme becomes equivalent to this random decision making mechanism with probability of detection being equal to the probability of false alarm. This value of Byzantine fraction is analogous to $\alpha_{blind}$ used in the literature \cite{rawat_tsp11}, which is the fraction of Byzantines needed to make the two hypotheses indistinguishable. For a typical distributed detection system, this value has been shown to be $0.5$ under the assumption of identical sensor decision rules. Our present results show that the FDR based scheme has a higher blinding fraction ($\alpha_{blind}=1$) which shows the robustness of the FDR based distributed detection framework.% and is now superior to the identical threshold scheme for all $\alpha$.

\begin{figure}[htb]
\centering
\includegraphics[width = 3.5in,height=!]{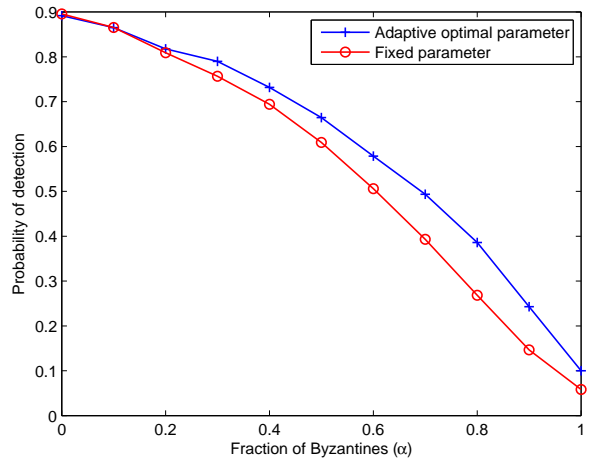}
\caption{Probability of detection against the fraction of Byzantines when the true hypothesis is $G_1$}
\label{Pd_alpha_pf_01_adap}
\end{figure}

%\textbf{NEED TO THINK ABT WAT TO DO HERE!!}Similarly, Figs. \ref{ROC_01} and \ref{ROC_04} shows the ROC for varying $\alpha$ showing the superiority of FDR-based scheme over identical scheme. 

%\begin{figure}[htb]
%\centering
%\includegraphics[width = 3.5in,height=!]{ROC_01}
%\caption{ROC using the optimal parameters for $\alpha=0.1$}
%\label{ROC_01}
%\end{figure}
%
%\begin{figure}[htb]
%\centering
%\includegraphics[width = 3.5in,height=!]{ROC_04}
%\caption{ROC using the optimal parameters for $\alpha=0.4$}
%\label{ROC_04}
%\end{figure}

Since we have shown that the optimal parameter value depends on the fraction of Byzantines present in the network, it becomes important to learn this parameter to adapt the system using the parameters at hand (local thresholds). We use a modified Kolmogorov-Smirnov Test, proposed in the following sub-section, to learn the fraction of Byzantines present in the network.

\subsection{Modified Kolmogorov-Smirnov Test}
\label{KStest}
Kolmogorov-Smirnov (K-S) test \cite{chakra67} is a goodness-of-fit test which compares a sample observed data with a reference null hypothesis distribution. It quantifies a distance metric (Kolmogorov-Smirnov distance) between the sample empirical cumulative distribution function (cdf) and the null hypothesis cdf to decide the goodness-of-fit. It is typically used only for continuous cdfs but Conover \cite{conover72} has extended this to cover the case of discontinuous cdfs. Since the only information known to the FC at every time instant is the count statistic ($\Delta$), a goodness-of-fit test on the count statistic has to be used to decide the range of $\alpha$. 

\subsubsection*{Description of the test \cite{conover72}}
Let $X_1$,  $X_2$,$\cdots$, $X_n$ represent a random sample of size $n$. Denote the null hypothesis by
\begin{equation}
K_0: F(x) = H(x) \qquad \text{for all $x$},
\end{equation}

where $F(x)$ is the unknown population distribution function, and $H(x)$ is the hypothesized distribution function with all parameters specified. $H(x)$ may be continuous, discrete, or a mixture of the two types. 
Let  $S_n(x)$ represent the empirical distribution function, 
\begin{equation}
S_n(x)  =  \frac{1}{n} (\text{the number of $X_i$'s which are $\leq x$}), \qquad \text{for all $x$}.
\end{equation}

The algorithm proposed by Conover \cite{conover72} gives a critical value of the sample which quantifies the confidence of null hypothesis being true. The test statistic $D$ is defined as 
\begin{equation}
D=\sup_x{|H(x)-S_n(x)|}
\end{equation}

In our case, since the optimal parameter value remains relatively fixed over different regions of $\alpha$, we can approximate it by a staircase function. If $Q$ is the number of regions of $\alpha$, we have a $Q$-ary hypothesis test. We compare the sample with the distribution of the count statistic under each of these regions

\begin{eqnarray}
K_i: P(\Delta;G_1)\qquad \text{for $\alpha=\alpha_{i}$}\quad \text{for $i=1,\cdots,Q$}
\end{eqnarray}

where $Q$ is the number of regions and $\alpha_i$ represents the $i^{th}$ region of $\alpha$.

Here, we modify the algorithm provided by Conover \cite{conover72} to generate test statistics $D_i$ for each of the hypotheses $(K_i$). Since we have the distribution of the count statistic for each of the hypotheses (using the analytical expressions derived in Section \ref{pmf}), we can find the test statistics under all the hypotheses using the given sample data. We then decide the hypothesis $H_i$ for which the $D_i$ is larger. The advantage of using the K-S test is that it performs well even with relatively small number of samples (e.g., 20-30 samples) compared to Pearson's Chi-Square test \cite{plackett83} which requires a larger number of samples.

%\textbf{where is the description of the test ???}

\subsection{Adaptive Algorithm}
In this subsection, we propose an adaptive algorithm based on the modified K-S test described above. In this algorithm, at every time instant $t$, the FC stores the count value of the previous $T_0$ time instants for which the global decision of $G_1$ was made. Using these $T_0$ data samples, it makes a decision regarding the region in which $\alpha$ lies using the modified K-S test described in Section \ref{KStest}. Depending on the decision made, it changes the detector parameter values. Let $K_i$ denote the decision made using the modified K-S test, then the detector parameter values are changed to $\gamma=\gamma_{i}, T=T_{i}, \kappa=\kappa_{i}$. The global threshold parameters $(T,\kappa)$ may also need to be changed in order to maintain the system-wide false alarm probability $P_{FA}$ at the desired value.

We now provide simulation results of our proposed adaptive algorithm. The system and target parameters are: $N=20$, $R=10$, $P_0=5$, $d_0=5$ and $P_{FA}=0.1$. This gives us the optimal FDR parameter values varying in $Q=2$ regions. The optimal parameter values are $\gamma_{1}=0.25$ and $\gamma_{2}=0.1$. For the K-S hypothesis test, we have used $T_0=30$ samples and the distributions for $\alpha_{1}=0$ and $\alpha_{2}=0.5$ under the two hypotheses have been found using (\ref{Eq:FDR_G1}) given in Section \ref{pmf}. The system is initially Byzantine free, i.e., $\alpha=0$. At $t=30$, $\alpha$ changes to $0.7$. In Fig. \ref{adaptive_0to7at30}, the global detection probability is plotted against time for the proposed adaptive algorithm and a non-adaptive algorithm which continues to use the initial detector parameters. As can be observed, the detection performance improves when the adaptive algorithm is used.

\begin{figure}[htb]
\centering
\includegraphics[width = 3.5in,height=!]{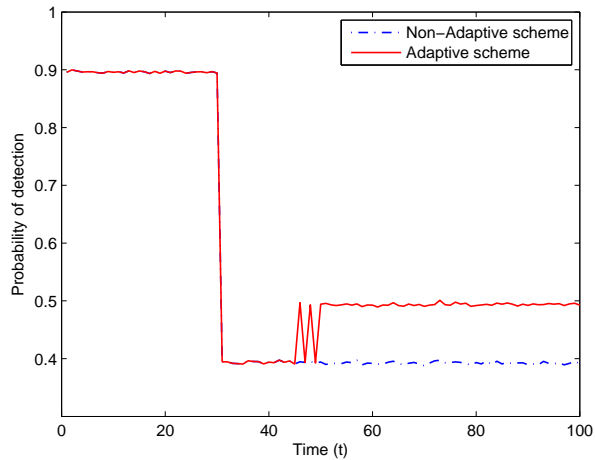}
\caption{Probability of detection versus time when $\alpha$ changes from 0 to 0.7 at $t=30$}
\label{adaptive_0to7at30}
\end{figure}

It is important to note that in the above example there were two regions for $\alpha$. However, the basic idea of learning $\alpha$ can be used when there are $Q$ regions for $\alpha$. 

\section{Discussion}
\label{conc}
In this section, we discuss the results presented in this paper with a focus on their implications and some implementation issues. We considered the problem of FDR based distributed detection in the presence of Byzantines. We first re-visited the work of \cite{ray_aes11} and observed that deflection coefficient is not the best heuristic for the design of FDR based distributed detection framework in non-asymptotic cases. Hence, we explored other possible distance measures that can better serve as design heuristics. Through empirical studies and analytical justifications, we showed that system performance can be improved by the use of Kolmogorov-Smirnov distance as the design heuristic. The advantage of using distance measures is the simplicity associated with its implementation in practice. The optimization is performed offline by finding the optimal parameter value through brute-force search as illustrated through simulations in this paper. The only computationally expensive step in evaluating the optimal parameter value is computation of count distribution under $G_0$ and $G_1$. Under $G_0$, it takes linear time to find the distribution since we have the closed form expressions derived in the paper. Under $G_1$, evaluating the distribution analytically involves multi-dimensional integration and, therefore, is of polynomial time complexity. Note that the distribution can also be approximated using Monte-Carlo runs. For a large number of sensors, we can use the asymptotic expressions derived in the paper which can be computed in linear time. Also, the distributed algorithm used in this framework reduces the energy and bandwidth consumption since the local sensors only report their 1-bit decisions to the FC rather than their p-values.

We also explored the system in the presence of malicious sensors (Byzantines). We modeled the Byzantines' attack strategy to ensure covertness in its behavior (since FDR value is still controlled at the pre-determined threshold $\gamma$) while degrading the system performance in terms of detection probability. We analyzed the system performance both theoretically and numerically. We observed that the optimal parameter value for the system depends on the fraction of Byzantines present in the system. The system performance degrades under severe attacks when fixed parameter values are used and, therefore, we proposed an adaptive approach to improve the performance which degraded in the presence of Byzantines. The sensors are deployed in a dynamically changing environment and, therefore, an adaptive scheme is necessary to combat the adversaries in the network. The proposed scheme learns the fraction of Byzantines present in the network and adaptively changes system parameter values to improve the global detection performance. 

There are several directions for future work on this problem. One could explore other distance measures \cite{basseville89}  that can be used as heuristics for system design. The optimal attack strategy for the Byzantines needs to be derived as it would be interesting to see how the performance of the network depends on the optimal attack strategy of the Byzantines defined by $h_{opt}(\cdot)$. Here, we considered the Neyman-Pearson framework. This work could be extended to a Bayesian framework where the problem is to detect the presence of a random target. For this, one may need to employ the Bayesian version of FDR called Bayesian FDR \cite{Whittemore07} or pFDR \cite{Storey03}. In \cite{Whittemore07}, an approach for control of Bayesian FDR has been proposed which can be used to design local sensor thresholds in distributed detection under the Bayesian framework similar to \cite{ray_aes11} and the present work.

\appendices
\section{Proof of Proposition \ref{prop_FDR_G1}}
\label{proof_FDR_G1}
For a Gaussian random variable $N(\phi,1)$, the pdf of a p-value $u$ of a true sensor observation located at a radial distance $r$ from the target is given by 
\begin{equation}
f(u;r)=\exp\left(-\frac{\phi^2}{2}\right)\exp(\phi Q^{-1}(u)),	\text{$0 \leq u \leq 1$}
\end{equation}
where $\phi$ = 0 if $r > d_0$ and $\phi$ = $\sqrt{P_0}$ if $r \leq d_0$.
The marginal pdf of the p-values in the presence of target can be found as 
\begin{align}
f_U(u)&=\int_0^{d_0}\exp\left(-\frac{P_0}{2}\right)\exp(\sqrt{P_0}Q^{-1}(u))f_R(r)dr + \int_{d_0}^{R}f_R(r)dr\\
&=\frac{d_0^2}{R^2}\exp\left(-\frac{P_0}{2}\right)\exp(\sqrt{P_0}Q^{-1}(u)) + \Bigg(1-\frac{d_0^2}{R^2}\Bigg)
\end{align}

where $f_R(r)=\frac{2r}{R^2}$ for $0 \leq r \leq R$ has been used.

This gives the marginal pdf of the transformed p-values $v$ as
\begin{equation}
\label{marginal_pdf}
f_V(v)=
\begin{cases}
f_U(v)	&	\text{if $i$ is an honest sensor}\\
f_U(1-v)	&	\text{if $i$ is a Byzantine sensor}
\end{cases}
\end{equation}

Under the assumption of independent p-values of the sensor observations, the transformed p-values $v$ are also independent. The FDR control algorithm requires ordering of these transformed p-values denoted by $v_{1,N}\leq v_{2,N} \cdots \leq v_{N,N}$ and are correlated (due to the ordering) with joint pdf given by
\begin{equation}
f_{V_{1,N}V_{2,N}\cdots V_{N,N}}=N!f_{V_1}(v_1)f_{V_2}(v_2)\cdots f_{V_N}(v_N)	\qquad \text{$0 \leq v_{1,N} \leq v_{2,N} \leq \cdots \leq v_{N,N}$}
\end{equation}
where the marginal density $f_{V_i}(v_i)$ is given by (\ref{marginal_pdf})

These ordered transformed p-values are compared against linearly decreasing thresholds to get the number of detections and, therefore, the probability can be found as
\begin{multline}
P(\Delta = i; G_1)= \int_{v_{N,N}=\gamma}^1 \cdots \int_{v_{i+1,N}=(i+1)\gamma/N}^{v_{i+2,N}}\int_{v_{i,N}=0}^{i\gamma /N}\\
\cdots \int_{v_{1,N}=0}^{v_{2,N}}N!f_{V_1}(v_1)f_{V_2}(v_2)\cdots f_{V_N}(v_N)dv_{1,N}dv_{2,N}\cdots dv_{N,N}
\end{multline}

However, since any of the $M$ sensors can be Byzantines, we need to take an average over the $\binom{N}{M}$ possibilities which gives the desired result.

For a large number of sensors, we can derive the approximate distribution of the count statistic $\Delta$ under $G_1$ using the result by Genovese and Wasserman \cite{genovese&wasserman02} which states that asymptotically the Benjamini-Hochberg method corresponds to classifying as $H_1$ all p-values that are less than a particular threshold $v^*$, where $v^*$ is the solution to the equation
\begin{equation}
F(v)=\beta v
\end{equation}
and 
\begin{equation}
\beta=\frac{\frac{1}{\gamma}-A_0}{1-A_0}
\end{equation}
Here $F(v)$ is the c.d.f of the transformed p-values under $H_1$ and is assumed to be strictly concave, and $A_0$ is the fraction of true $H_0$s. This threshold $v^*$ is found by assuming a mixture model of the distribution of p-values. For honest sensors, $F_H(v)=Q(Q^{-1}(v)-\phi)$ and for Byzantine sensors $F_B(v)=1-Q(Q^{-1}(1-v)-\phi)$, where $\phi=\sqrt{P_0}$. So, under the mixture model, we have $F(v)=\alpha F_B(v)+(1-\alpha)F_H(v)$. For a large ROI, on an average $d_0^2/R^2$ fractions of sensors receive the signal, and therefore $A_0=1-d_0^2/R^2$.

Hence, the average probability of detection of an honest sensor is given by
\begin{align}
\bar{p}_D^H&=(1-A_0)P(V<v^*|H_1)+A_0P(V<v^*|H_0)\\
&=(1-A_0)\int_0^{v^*}f_\phi(u)du + A_0v^*
\end{align}
where $f_\phi(u)$ is given by %(\ref{fu}).

\begin{equation}
\label{fu}
f_\phi(u)=\exp\left(-\frac{\phi^2}{2}\right)\exp(\phi Q^{-1}(u)), \qquad \text{$0 \leq u \leq 1$}
\end{equation}

Similarly, the probability of detection of a Byzantine sensor is given by

\begin{align}
\bar{p}_D^B&=(1-A_0)P(V<v^*|H_1)+A_0P(V<v^*|H_0)\\
&=(1-A_0)P(1-U<v^*|H_1)+A_0P(1-U<v^*|H_0)\\
&=(1-A_0)\int^1_{1-v^*}f_\phi(u)du + A_0v^*
\end{align}

The probability of $\Delta =i$ detections (counts) when the target is present is provided by
\begin{equation}
P(\Delta = i; G_1)= \sum_{k=\max{(0,M-N+i)}}^{\min{(M,i)}}\binom{M}{k}\binom{N-M}{i-k}(\bar{p}_D^B)^k(1-\bar{p}_D^B)^{M-k}(\bar{p}_D^H)^{i-k}(1-\bar{p}_D^H)^{(N-M)-(i-k)}
\end{equation}
or 
\begin{equation}
P(\Delta = i; G_1)= \binom{N}{i}(\bar{p}_D)^i(1-\bar{p}_D)^{N-i}
\end{equation}
where $\bar{p}_D$ is the average probability of detection of a sensor given by
\begin{equation}
\bar{p}_D=\alpha \bar{p}_D^B + (1-\alpha) \bar{p}_D^H
\end{equation}

Also, we can further approximate this expression using DeMoivre-Laplace theorem as a Gaussian distribution
\begin{equation}
P(\Delta =i;G_1)\approx \frac{1}{\sqrt{2\pi N\bar{p}_D(1-\bar{p}_D)}}\exp(\frac{(i-N\bar{p}_D)^2}{2 N\bar{p}_D(1-\bar{p}_D)})
\end{equation}

%\section{Proof of Proposition \ref{prop_ident_G1}}
%\label{proof_ident_G1}
%
%The average probability of detection for an honest sensor is given by
%
%\begin{align}
%\bar{p}_{D_{IT}}^H&=(1-A_0)P(V<p_{fa}|H_1)+A_0P(V<p_{fa}|H_0)\\
%&=(1-A_0)P(U<p_{fa}|H_1)+A_0P(U<p_{fa}|H_0)\\
%&=(1-A_0)\int_0^{p_{fa}}f_\phi(u)du + A_0 p_{fa}
%\end{align}
%
%where $f_\phi(u)$ is given by (\ref{fu}).
%
%Similarly, the average probability of detection for a Byzantine sensor is given by
%
%\begin{align}
%\bar{p}_{D_{IT}}^B&=(1-A_0)P(V<p_{fa}|H_1)+A_0P(V<p_{fa}|H_0)\\
%&=(1-A_0)P(1-U<p_{fa}|H_1)+A_0P(1-U<p_{fa}|H_0)\\
%&=(1-A_0)\int_{1-p_{fa}}^1f_\phi(u)du + A_0 p_{fa}
%\end{align}
%
%Now, the probability of $\Delta =i$ under $G_1$ is given by
%\begin{equation}
%P(\Delta =i;G_1)=\sum_{k=\max{(0,M-N+i)}}^{\min{(M,i)}}\binom{M}{k}\binom{N-M}{i-k}(\bar{p}_{D_{IT}}^B)^k(1-\bar{p}_{D_{IT}}^B)^{M-k}(\bar{p}_{D_{IT}}^H)^{i-k}(1-\bar{p}_{D_{IT}}^H)^{(N-M)-(i-k)}
%\end{equation}
%
%or 
%\begin{equation}
%P(\Delta =i;G_1)=\binom{N}{i}(\bar{p}_{D_{IT}})^i(1-\bar{p}_{D_{IT}})^{N-i}
%\end{equation}
%
%where $\bar{p}_{D_{IT}}$ is the average probability of detection of a sensor given by
%\begin{equation}
%\bar{p}_{D_{IT}}=\alpha \bar{p}_{D_{IT}}^B + (1-\alpha)\bar{p}_{D_{IT}}^H
%\end{equation}
%
%As before, it can be further approximated using the DeMoivre-Laplace theorem as
%\begin{equation}
%P(\Delta =i;G_1)\approx \frac{1}{\sqrt{2\pi N\bar{p}_{D_{IT}}(1-\bar{p}_{D_{IT}})}}\exp(\frac{(i-N\bar{p}_{D_{IT}})^2}{2 N\bar{p}_{D_{IT}}(1-\bar{p}_{D_{IT}})})
%\end{equation}

\bibliographystyle{IEEEtran}
\bibliography{Conf,Book,Journal}
\end{document}